\documentclass{article}
\usepackage{spconf,amsmath,graphicx}
\usepackage{amsthm,amsfonts,amssymb}
\usepackage{mathbbol}
\usepackage{multirow}
\usepackage{xcolor}
\usepackage{tikz-cd}
\usepackage{tikz}
\usepackage{url}
\usepackage{amsmath}
\usepackage{amsfonts}
\usepackage{caption}
\usepackage{colortbl}
\usepackage{booktabs}
\usepackage{subcaption}

\theoremstyle{plain}
\newtheorem{Theorem}{Theorem}

\newtheorem{Proposition}[Theorem]{Proposition}

\definecolor{darkblue}{RGB}{0,0,100}	
\definecolor{darkred}{RGB}{153,0,0}	
\definecolor{darkredc}{RGB}{90,0,0}
\definecolor{thisblue}{rgb}{0.03, 0.27, 0.49}
\definecolor{darkgreen}{rgb}{0.0, 0.42, 0.24}
\definecolor{myorange}{rgb}{0.8, 0.5, 0.4}

\title{Overcoming Measurement Inconsistency in Deep Learning for Linear Inverse
Problems: Applications in Medical Imaging}
\name{
  Marija Vella 
  \qquad\qquad 
  Jo\~{a}o F. C. Mota
  \thanks{
    Work supported by EPSRC's New Investigator Award EP/T026111/1 and by the Capital
    Award EP/S018018/1.
  }
}
\address{Institute of Sensors, Signals and Systems, \,Heriot-Watt University,\, Edinburgh,\, UK}
\newcommand{\mypar}[1]{{\bf #1.}}

\begin{document}
\ninept
\maketitle
\begin{abstract}
The remarkable performance of deep neural networks (DNNs) currently makes them the method of choice for solving linear inverse problems. They have been applied to super-resolve and restore images, as well as to reconstruct MR and CT images. In these applications, DNNs invert a forward operator by finding, via training data, a map between the measurements and the input images. It is then expected that the map is still valid for the test data. This framework, however, introduces measurement inconsistency during testing. We show that such inconsistency, which can be critical in domains like medical imaging or defense, is intimately related to the generalization error. We then propose a framework that post-processes the output of DNNs with an optimization algorithm that enforces measurement consistency. Experiments on MR images show that enforcing measurement consistency via our method can lead to large gains in reconstruction performance.
\end{abstract}
\section{Introduction}
\label{sec:intro}

Many applications in science and engineering require solving inverse problems
in which the number of available measurements is much smaller than the number
of parameters to be estimated. Examples include various medical imaging
modalities, remote sensing, image restoration, seismography, and LiDAR depth
estimation.  In such problems, we have access to $m$ linear measurements of a
vector $x^\star \in \mathbb{C}^n$ that we wish to estimate.  Formally, we have
the linear system
\begin{equation}
b = Ax^\star\,,
\label{eq:ILP}
\end{equation}
where $b \in \mathbb{C}^{m}$ is the vector of measurements, and $A \in
\mathbb{C}^{m \times n}$ is a known measurement matrix with $m < n$.
As the number of measurements in~\eqref{eq:ILP} is smaller than the number of
variables, there is often an infinite number of vectors $x \in \mathbb{C}^n$
satisfying~\eqref{eq:ILP}. 

The classical approach of inferring $x^\star$ from~\eqref{eq:ILP} formulates an
optimization problem for finding the simplest solution of~\eqref{eq:ILP}
according to the known structure of $x^\star$. Such structure is captured by
regularizers such as the
$\ell_1$-norm~\cite{Donoho98-AtomicDecompositionBasisPursuit}, which enforces
sparsity in a given domain, or generalizations like the total variation (TV)
norm~\cite{Rudin92-NonlinearTotalVariationBasedNoiseRemovalAlgorithms} and
structural or hierarchical
priors~\cite{Baraniuk2010modelbasedcompressedsensing}. In recent years these
optimization-based methods for linear inverse problems have been surpassed by
data-driven approaches, specifically, deep (convolutional) neural networks
(DNNs). By leveraging large datasets of input-output pairs $(x^\star,\, b)$
during training, DNNs are able to automatically learn the structure of typical
signals. During deployment, this enables them not only to reconstruct $x^\star$
with quality better than optimization approaches, but also to do it faster.
This phenomenon has been observed in several linear inverse problems, including
single-image
super-resolution~\cite{Dong16-ImageSuperResolutionUsingDeepConvolutionalNetworks},
denoising~\cite{zhang2018ffdnet}, biomedical
imaging~\cite{Jin17-DeepConvolutionalNeuralNetworkForInverseProblemsInImaging},
and LiDAR depth estimation~\cite{Uhrig17-SparsityInvariantCNNs}.

Despite these successes, DNNs still suffer from important
drawbacks that have slowed down their application in critical domains, such as
autonomous driving or fully automated medical diagnosis. These include
overfitting, lack of sharp theoretical guarantees, and
instability with respect to small data
perturbations~\cite{Antun2020insstabilities}. A related drawback, which is
particularly important in linear inverse problems and which we explore in this
paper, is measurement inconsistency. 

\mypar{Measurement inconsistency in DNNs}
Typically, DNNs for linear inverse problems are trained by minimizing a
real-valued loss function $\ell\,:\, \mathbb{C}^{n} \times \mathbb{C}^n \to
\mathbb{R}$, usually the $\ell_2$-norm, over all the $T$ samples
$\{x^{(t)}\}_{t=1}^T$ of the training set. That is, if
$f_{\theta}\, :\, \mathbb{C}^m \to \mathbb{C}^n$ represents a DNN with
parameters $\theta$, one finds the optimal set of parameters $\theta^\star$ by
solving
\begin{equation}
  \label{Eq:TrainingACNN}
  \underset{\theta}{\text{minimize}}\,\,\, \frac{1}{T}\sum_{t = 1}^{T} \ell\Big(x^{(t)},
  f_{\theta}(A x^{(t)})\Big)\,.
\end{equation}
Once trained, given a vector of measurements $b$
from a sample $x^\star$ absent from the training set, the
DNN estimates $x^\star$ by simply applying a forward pass to $b$:
\begin{align}
  \label{Eq:ForwardPass}
  w := f_{\theta^\star}(b)\,.
\end{align}
We say that the trained DNN $f_{\theta^\star}$ is \textit{measurement
  inconsistent} when $Aw
  \neq b$.\footnote{\label{foot:noise}This definition can be extended to the case in which the model in~\eqref{eq:ILP} is noisy, $b = Ax^\star + e$, and $e$ is
    bounded: $\|e\|_2 \leq
  \sigma$, where $e \in \mathbb{C}^m$ is the noise vector. In this case, we say the DNN is measurement inconsistent when $\|Aw
- b\|_2 > \sigma$.} In fact, we show in Proposition~\ref{Prop:Bound} that
measurement inconsistency is a by-product of generalization errors, even when the DNN is
trained to minimize inconsistency, i.e., when the loss function
in~\eqref{Eq:TrainingACNN} is $\ell(x,\, w) := \|Ax - Aw\|_2^2$. 

This calls for a new framework that not only enforces data consistency,
but also can harness the excellent performance of DNNs. And while the data
consistency problem has been identified as an important one, for example, in medical imaging, the existing solutions consist mostly of modifying the network
architecture. For instance, \cite{Jin2017inverse,lim2017edsr} introduce skip-connections between the input and
output, and \cite{zhang2018learning} trains a DNN using different
models of $A$. These solutions still fall under the
framework of~\eqref{Eq:TrainingACNN} and, thus, suffer from the problem
highlighted in Proposition~\ref{Prop:Bound}. Our approach consists instead of
post-processing the output of the DNN, $w$ in~\eqref{Eq:ForwardPass}, with an
optimization algorithm that enforces consistency.

\begin{figure}[ht]
	\includegraphics[width=8.5cm, height=3.4cm]{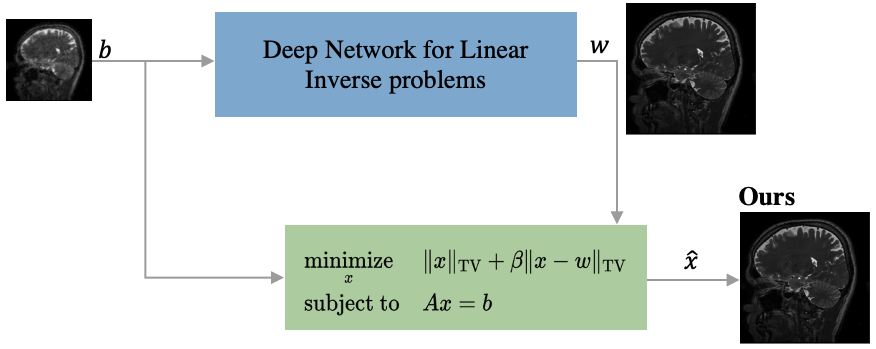}
	\caption{Our framework. Measurement consistency is enforced by
  post-processing the output of a DNN designed for a linear inverse problem
with an optimization problem with $Ax = b$ as a constraint.}
	\label{fig:frameworkV2}
\end{figure}

\mypar{Our approach and contributions}
Fig.~\ref{fig:frameworkV2} shows a diagram of the scheme we propose. 
The scheme addresses measurement consistency by formulating an
optimization problem, named \textit{TV-TV minimization}, that
explicitly enforces consistency via constraints, while minimizing a cost
function that balances a small total variation (TV) of the reconstructed image ---
as in classical optimization-based approaches --- and proximity to the output
of the DNN, as measured by the TV-norm. This last term is the key element of
our approach: it allows combining the major benefit of optimization-based
methods, i.e., the ability to enforce constraints explicitly, with the major
benefit of data-driven methods, i.e., their outstanding performance. Our
experiments show that our framework improves the reconstruction performance
of state-of-the-art DNNs for MRI reconstruction, e.g.,
MoDL~\cite{aggarwal2019modl}, by $5\, \text{dB}$ in PSNR. 
We summarize our contributions as follows:
\begin{itemize}

  \item We show that the conventional method for training DNNs according to
    empirical loss minimization~\eqref{Eq:TrainingACNN} cannot solve the
    measurement consistency problem in linear inverse problems; this is done in Proposition~\ref{Prop:Bound} and
    illustrated with experiments.

  \item We propose a framework that addresses the measurement consistency
    problem by combining an optimization-based method, TV-TV minimization, with
    DNNs.

  \item We apply our framework to MRI and show that it can
    lead to significant gains with respect to state-of-the-art DNNs.
\end{itemize}
We point out that we proposed a particular instance of this framework
in~\cite{vella2019tvtvminimization,Vella20-RobustSingleImageSuperResolutionViaCNNsAndTVTVMinimization}
for single-image super-resolution. There, the observed gains in PSNR were
much smaller ($<1 \, \text{dB}$) than the ones observe here for MRI. In this
paper, we address the consistency problem theoretically. And to 
apply TV-TV minimization to MRI, we had to generalize the algorithm to handle
complex vectors and matrices. 

After overviewing prior work in Section~\ref{Sec:RelatedWork}, we present our
theoretical result on consistency in Section~\ref{Sec:Inconsistency}.
Section~\ref{Sec:Framework} describes our framework, and
Section~\ref{sec:experiments} our experiments on MRI reconstruction.

\section{Related Work}
\label{Sec:RelatedWork}

\mypar{Optimization-based approaches}
Inverse linear problems have traditionally been solved via optimization algorithms. The associated optimization problem is formulated such that its solution simultaneously conforms to the observed measurements and is sufficiently simple according to some prior knowledge. The latter is usually encoded by sparsity in some domain, such as wavelet or DCT representations~\cite{Donoho98-AtomicDecompositionBasisPursuit}, or in gradient space~\cite{Rudin92-NonlinearTotalVariationBasedNoiseRemovalAlgorithms}. Sparsity can be encoded directly via an $\ell_0$-pseudo-norm, which leads to nonconvex problems that can be solved via greedy approaches~\cite{Donoho2012,needell2009coasmp}, or via a convex relaxation such as the $\ell_1$-norm, which leads to convex problems that can be provably solved~\cite{Boyd04-ConvexOptimization}. Convex formulations also usually have strong theoretical reconstruction guarantees~\cite{Chandrasekaran12-ConvexGeometryLinearInverseProblems}. Indeed, \cite{Chandrasekaran12-ConvexGeometryLinearInverseProblems} generalizes the concept of simplicity via sparsity to \textit{atomic norms}, which apply to a wider range of problems.
In the field of MRI reconstruction, a common choice for simplicity is sparsity in the gradient domain, which is captured by a total variation (TV) norm~\cite{Rudin92-NonlinearTotalVariationBasedNoiseRemovalAlgorithms,Block2007mritv,Weizman16-ReferenceBasedMRI} or generalizations of the TV-norm~\cite{knoll2012mritv}.

\mypar{DNN-based methods}
Deep neural networks (DNNs) have been successfully applied in various linear inverse problems~\cite{Dong16-ImageSuperResolutionUsingDeepConvolutionalNetworks,zhang2018ffdnet,Jin17-DeepConvolutionalNeuralNetworkForInverseProblemsInImaging,Uhrig17-SparsityInvariantCNNs,lim2017edsr,schlemper2018deepcascade}. These networks learn to invert a forward model
by leveraging several input-output pairs obtained from a single measurement operator. Once trained, they can be used to reconstruct the input of the operator from its output by a simple forward pass. This operator, however, often fails to guarantee measurement consistency. Attempts to solve this problem include adding skip-connections to share information from the input to the final layer~\cite{Jin2017inverse,lim2017edsr,Gupta2018}, embedding a data consistency layer in the network~\cite{schlemper2018deepcascade,qincrnn} and using a cycle consistency loss \cite{Oh202unpaireddeeplearning}. 
Another approach is to unroll iterative optimization algorithms that alternate between data consistency and a nonlinear operation related to prior knowledge. This idea was first proposed~\cite{Gregor2010}, and it spawned different lines of research, e.g., \cite{Chang_OneNet,aggarwal2019modl,Luong20-InterpretableDeepRecurrent}. For example,  \cite{aggarwal2019modl} trains a DNN denoiser whose output is then fed to an optimization block that acts as a data consistency layer~\cite{aggarwal2019modl}. As our experiments show, this is still not enough to guarantee consistency.

\mypar{Instability in DNNs}
In recent years, several issues have been identified in the deployment of DNNs, including overfitting and data memorization~\cite{Zhang17-UnderstandingDeepLearningRequiresRethinkingGeneralization}, and adversarial examples~\cite{DeepFool}. For linear inverse problems, \cite{Antun2020insstabilities} studied how small perturbations to the measurement or to the sampling method can lead to artefacts in the reconstruction. Moreover, the reconstructed outputs may miss fine details that are present in measurements, a feature that can have critical consequences in medical imaging. 

To overcome this, we present a framework that leverages the good performance of DNNs, while enforcing measurement consistency. Consequently, the risk of missing such details is reduced. In contrast to other methods, we enforce hard constraints in order to ensure measurement consistency. Before presenting our framework, we first see how inconsistency in DNNs is related to the generalization error.

\section{The probability of inconsistency in DNNs}
\label{Sec:Inconsistency}

Here we adopt a probabilistic setting to analyze the measurement
inconsistency problem of DNNs.  
Let $X \in \mathbb{C}^n$ represent a vector of $n$ random complex variables
whose $i$th component is $X_i = \text{Re}\{X_i\} + j\text{Im}\{X_i\}$. To avoid
technicalities, we assume all the functions we deal with are measurable,
including the DNN $f_{\theta}\,:\, \mathbb{C}^m \to\mathbb{C}^n$, for any
$\theta$. Without loss of generality, we assume the following squared
$\ell_2$-norm loss: $\ell(x,\, w) := \|Ax - Aw\|_2^2$. Hence, the \textit{expected
loss} associated to $f_{\theta}$ is defined as
\begin{equation}
  \label{Eq:ExpectedLoss}
  \ell_{\text{exp}}(f_{\theta}) 
  := 
  \mathbb{E}\Big[ \big\|AX - Af_{\theta}(A X)\big\|_2^2\Big]\,,
\end{equation}
where expectation is with respect to $X$. The \textit{empirical loss of}
$f_{\theta}$ on the training set $\mathcal{T} := \{x^{(t)}\}_{t=1}^T$, where
$x^{(t)}$ is a realization of $X$, is
\begin{equation}
  \label{Eq:EmpiricalLoss}
  \ell_{\text{emp}}\big(f_{\theta}\, ;\, \mathcal{T}\big) 
  := 
  \frac{1}{T}\sum_{t \in \mathcal{T}} \big\|Ax^{(t)} -
  Af_{\theta}(Ax^{(t)})\big\|_2^2\,.
\end{equation}
Notice that these definitions differ from conventional definitions of
expected and empirical loss in two ways (see, e.g.,
\cite{Amjad18-OnDeepLearningForInverseProblems,Luong20-InterpretableDeepRecurrent}
for examples in linear inverse problems). First, the probability distribution
is defined over the output space $X$ only. Indeed, under the assumption that
the measurements in~\eqref{eq:ILP} are noiseless, the input random variable $B
\in \mathbb{C}^m$ is completely specified by $X$: $B = AX$. Second, the
$\ell_2$-norm is applied to the output rather than to the input space: that is,
in~\eqref{Eq:ExpectedLoss}, we consider $\|AX - Af_{\theta}(AX)\|_2^2$ rather
than $\|X - f_{\theta}(AX)\|_2^2$. This
reflects a training strategy to minimize inconsistency, and captures the type of
regularization terms used in deep prior or network unrolling models, e.g.,
MoDL~\cite{aggarwal2019modl}. This is without loss of generality, as the loss can include additional terms. The following
proposition bounds the probability of the DNN
outputting an inconsistent result as a function of the generalization error
$c:= \ell_\text{exp}(f_{\theta}) - \ell_{\text{emp}}(f_{\theta}\,;\,
\mathcal{T})$. 
\begin{Proposition}
  \label{Prop:Bound}
  Consider $f_{\theta^\star}$ with a parameter $\theta^\star$ that achieves an
  empirical loss $\epsilon := \ell_{\text{emp}}(f_{\theta^\star}\,;\, \mathcal{T}) > 0$.
  Assume the random variable $Y:=\|AX -
  Af_{\theta^\star}(AX)\|_2^2$ is upper bounded by $C$ almost surely. Assume
  a positive generalization error $c := \ell_{exp}(f_{\theta}) - \epsilon
  > 0$. Then, for any $\delta$ such that $0 < \delta < c + \epsilon$, 
  \begin{equation}
    \label{Eq:Bound}
    \mathbb{P}\Big(\big\|AX - Af_{\theta^\star}(AX)\big\|_2^2 \geq \delta\Big)
    \geq
    1 - \exp\Big(-2\frac{(c + \epsilon - \delta)^2}{C^2}\Big)\,.
  \end{equation}
\end{Proposition}
\begin{proof}
  By assumption, $0 \leq Y \leq C$ almost surely, which implies that $Y$ is
  sub-Gaussian with parameter $\sigma := C/2$~\cite{Wainwright19-HDStats}. 
  Therefore, for any $t > 0$,
  \begin{align}
    \mathbb{P}(Y < \delta)
    &=
    \mathbb{P}\big(Y - \mathbb{E}[Y]< \delta - \mathbb{E}[Y]\big)
    \label{Eq:ProofS1}
    \\
    &=
    \mathbb{P}\big(Y - \mathbb{E}[Y] < \delta - c - \epsilon\big)
    \label{Eq:ProofS2}
    \\
    &\leq
    \mathbb{P}\big(|Y - \mathbb{E}[Y]| >  c + \epsilon - \delta\big)
    \label{Eq:ProofS3}
    \\
    &\leq
    \exp\Big(-t(c + \epsilon - \delta) + \frac{t^2 \sigma^2}{2}\Big)\,.
    \label{Eq:ProofS4}
  \end{align}
  From~\eqref{Eq:ProofS1} to~\eqref{Eq:ProofS2}, we used the definition of $c =
  \mathbb{E}[Y] - \epsilon$. From~\eqref{Eq:ProofS2} to~\eqref{Eq:ProofS3}, we
  used the fact that $\delta < c + \epsilon$. And from~\eqref{Eq:ProofS3}
  to~\eqref{Eq:ProofS4}, we applied a Chernoff
  bound~\cite[\S4.2]{Mitzenmatcher05-ProbabilityAndComputing} taking into
  account that $Y$ is sub-Gaussian with parameter
  $\sigma$~\cite[\S2.2.1]{Wainwright19-HDStats}. Setting $t = (c + \epsilon -
  \delta)/\sigma^2$, replacing $\sigma = C/2$, and taking the complementary
  event in~\eqref{Eq:ProofS4} yields~\eqref{Eq:Bound}.
\end{proof}
The left-hand side of~\eqref{Eq:Bound} expresses the probability that the
output of the DNN is inconsistent, and the right-hand side increases as a
function of the generalization gap $c$ as well as the empirical loss
$\epsilon$. The assumption that the random variable $Y = \|AX -
  Af_{\theta^\star}(AX)\|_2^2$ is bounded above can be easily relaxed: the
  proof can be adapted if we assume Y is sub-Gaussian or even sub-exponential
  (fat tails).

\section{Our Framework}
\label{Sec:Framework}

To address the measurement inconsistency problem studied in Section \ref{Sec:Inconsistency}, we propose the framework represented in Fig.~\ref{fig:frameworkV2}, which in contrast to conventional DNNs, is able to ensure measurement consistency.
In order to understand how it differs from previous approaches, recall
that conventional DNN-based methods for linear inverse problems use empirical loss
minimization~\eqref{Eq:TrainingACNN} to find a good-enough parameter
$\theta^\star$ of a DNN $f_{\theta}:
\mathbb{C}^m \to \mathbb{C}^n$. Then, during testing, they simply apply
$w:=f_{\theta^\star}(b)$ to measurements $b = Ax^\star$ of unseen data
$x^\star$.
This framework applies to DNNs 
designed for a specific inverse problem, e.g., single-image
super-resolution~\cite{Dong16-ImageSuperResolutionUsingDeepConvolutionalNetworks},
as well as unrolled networks, e.g., \cite{Gregor2010,Luong20-InterpretableDeepRecurrent}.
However, as shown in Proposition~\ref{Prop:Bound},
a generalization error typically implies inconsistency of the DNN output to the
measurements, i.e., $Af_{\theta^\star}(b) \neq b$.

\mypar{TV-TV minimization}
To overcome this problem, and as shown in Fig.~\ref{fig:frameworkV2}, we
feed the measurements $b$ together with the output $w = f_{\theta^\star}(b)$ of a given
DNN, which is generally very close to the desired $x^\star$, but measurement
inconsistent, into an optimization problem that we call TV-TV minimization:
\begin{equation}
\begin{array}[t]{ll}{
	\underset{x}{\text{minimize}}} & \|x\|_{\text{TV}}+\beta\|x-w\|_{\text{TV}} 
\\ 
  \text{subject to} & Ax=b\,.
\end{array}
\label{eq:TV-TV}
\end{equation}
In the objective function, $\|\cdot\|_{\text{TV}}$ stands for the 2D TV
semi-norm of an image $\overline{x} \in \mathbb{C}^{M \times N}$ [whose
vectorization is $x \in \mathbb{C}^n$ with $n = M\cdot N$]. It is defined as
$\|x\|_{\text{TV}} := \sum_{i=1}^{M}\sum_{j=1}^{N} |v_{ij}^\top x| + |h_{ij}^\top x| =
\|D x\|_1$,
where $v_{ij}, h_{ij} \in \mathbb{R}^n$ are real vectors that extract the
vertical and horizontal differences at pixel $(i, j)$ of $\overline{x}$, and $D
\in \mathbb{R}^{2n \times n}$ is the vertical concatenation of $v_{ij}^\top$
and $h_{ij}^\top$ for all $i = 1, \ldots, M$ and $j = 1, \ldots, N$. The first
term, $\|x\|_{\text{TV}}$, encodes the assumption that the image to reconstruct
has a small number of edges compared to its dimension. This is a standard
approach in optimization-based methods for image
restoration~\cite{Rudin92-NonlinearTotalVariationBasedNoiseRemovalAlgorithms,li2013efficient}
and MRI
reconstruction~\cite{Weizman16-ReferenceBasedMRI}. The second term, $\beta\|x -
w\|_{\text{TV}}$, specifies that the solution of~\eqref{eq:TV-TV} should be
close to the output $w$ of the DNN, in a TV-norm sense.  Here, $\beta$ balances
between the two terms of the objective. According to the theory
in~\cite{mota2019cswithpi} and the experiments
in~\cite{vella2019tvtvminimization,Vella20-RobustSingleImageSuperResolutionViaCNNsAndTVTVMinimization},
a value close to $\beta = 1$ yields the best results. Finally, consistency is
achieved by constraining the solution of~\eqref{eq:TV-TV} to satisfy $Ax = b$.
Notice that~\eqref{eq:TV-TV} can be modified to accommodate noisy measurements
as in footnote~\ref{foot:noise}. However, as will be shown in Section~\ref{sec:experiments},
we found that even when considering a noiseless (and thus inaccurate) model,
post-processing the output $w$ of a state-of-the-art DNN for MRI reconstruction
via~\eqref{eq:TV-TV} leads to significant performance gains. 

\mypar{Algorithm for solving~\eqref{eq:TV-TV}}
We apply ADMM~\cite{boyd_distributed_2011} to a reformulation
of~\eqref{eq:TV-TV}. The main idea and algorithm are described
in~\cite{Vella20-RobustSingleImageSuperResolutionViaCNNsAndTVTVMinimization}.
But to apply the resulting algorithm to MRI reconstruction, we had to slightly
change the reformulations in order to handle complex vectors and matrices. Details will be described in a forthcoming paper.

 \begin{figure*}[ht]
 	\centering
 	\begin{subfigure}[b]{0.3\textwidth}
 		\centering
 		\includegraphics[width=5.3cm, height=6cm]{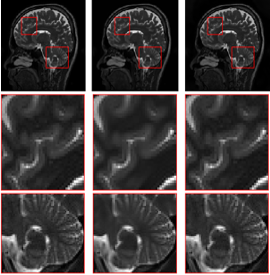}
 		\smallskip
    \hspace{0cm}{GT}\hspace{1.2cm}{MoDL}\hspace{1cm}{\textbf{\textit{\textcolor{darkred}{Ours}}}}
        \caption{Result on a sample image using MoDL.}
        \label{SubFig:MoDLA}
 	\end{subfigure}
 	\hfill
 	 	\begin{subfigure}[b]{0.3\textwidth}
 		\centering
 		\includegraphics[width=5.3cm, height=6cm]{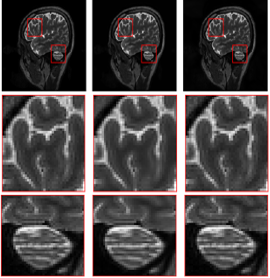}
 		\smallskip
        \hspace{0cm}{GT}\hspace{1.3cm}{MoDL}\hspace{1cm}{\textbf{\textit{\textcolor{darkred}{Ours}}}}
        \caption{Result on a sample image using MoDL.}
        \label{SubFig:MoDLB}
 	\end{subfigure}
 	\hfill
  	\begin{subfigure}[b]{0.3\textwidth}
 	    \centering
 	    \includegraphics[width=5.3cm, height=6cm]{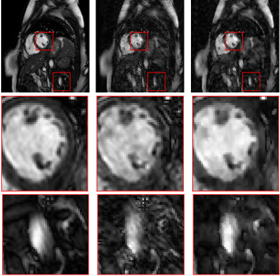}
 		\smallskip
        \hspace{0.1cm}{GT}\hspace{1.2cm}{CRNN}\hspace{1cm}{\textbf{\textit{\textcolor{darkred}{Ours}}}}
        \caption{Result on a sample image using CRNN.}
        \label{SubFig:CRNN}
 \end{subfigure}
 \caption{Reconstruction of two test images \text{(a)}-\text{(b)} for MoDL~\cite{aggarwal2019modl}, and \text{(c)} one test image for CRNN~\cite{qincrnn}. Each figure shows the full image in the top row, and two zoomed regions in the two bottom rows; the left column shows the groundtruth (GT), and the right one our method.}
 \label{fig:results}
\end{figure*}


\section{Application to MRI Reconstruction}
\label{sec:experiments}

We now describe our experiments for MRI reconstruction.
After explaining the setup, we illustrate how the proposed framework solves the
measurement consistency problem and how this leads to significant gains in
reconstruction performance. Code to replicate our experiments is available
online\footnote{https://github.com/marijavella/mri-tvtv}.

\mypar{Experimental setup}
We applied our framework to two state-of-the-art MRI
reconstruction DNNs:  MoDL~\cite{aggarwal2019modl} and CRNN~\cite{qincrnn}.
MoDL reconstructs multichannel MRI images in which data is acquired by a
12-channel head coil, and was trained on the multichannel brain dataset acquired by the authors of~\cite{aggarwal2019modl}. For testing, we used 164 slices from a single subject,
resulting in a test dataset of dimensions $256\times 232 \times 164 \times 12$
($\text{rows} \times \text{columns} \times  \text{slices} \times \text{number of
coils}$). CRNN, in turn, applies to data acquired by a single-channel coil, and
was trained on cardiac images from various subjects \cite{qincrnn}. To
avoid retraining CRNN, we used a pretrained version of the network, which was
trained for a single subject. As we will see, this led to low quality outputs.
Both methods process the real and imaginary parts of the complex MRI data
separately. We set an acceleration factor of 6 for MoDL and of 4 for CRNN.
During testing, we added no artificial noise to the data.

For MoDL (resp.\ CRNN), we set $\beta=1$ (resp.\ $\beta = 0.8$)
in~\eqref{eq:TV-TV} and ran our algorithm a maximum number 
of 100 (resp.\ 50) ADMM iterations. The measurement matrix $A$ in~\eqref{eq:ILP} for MoDL was the product of a sampling mask $S$, a coil sensitivity map $C$,
and a 2D discrete Fourier transform $F$, i.e., $A = SFC$, and for CRNN it was just $A = SF$, as this
network operates on a single coil only.

Because background noise in the images introduced some variation
in the results, we evaluated the performance metrics, PSNR and SSIM, on
images cropped to the relevant anatomic content.

\begin{table}
  \footnotesize
	\centering
	\renewcommand{\arraystretch}{1}
	\caption{
    Measurement consistency of MoDL, 
    CRNN (2nd column), and corresponding consistency of our method (3rd column).
	}
	\begin{tabular}{@{}lll@{}}
    \toprule
		\textcolor{thisblue}{\textbf{Method}} 
    &
    \textcolor{thisblue}{$\boldsymbol{\|Aw-b\|_2}$} 
    &
    \textcolor{darkredc}{$\boldsymbol{\|A\hat{x}-b\|_2}$} \\
		\midrule
    \textcolor{darkgreen}{MoDL} \cite{aggarwal2019modl} 
		 & $3.10\times10^{-1} $ & $\boldsymbol{9.88\times10^{-5}}$\\
		\midrule
    \textcolor{myorange}{CRNN}~\cite{qincrnn}
		& $2.06\times10^{-6}$  & $\boldsymbol{7.71\times10^{-15}}$ \\
		\bottomrule
	\end{tabular}
	\label{tab:consistency}
\end{table}

\mypar{Measurement consistency}
Table~\ref{tab:consistency} displays the consistency metric $\|Ax - b\|_2$ for the outputs of MoDL and CRNN (2nd column), and the respective metric after post-processing with our algorithm (3rd column). The input and output images are displayed in Figs.~\ref{SubFig:MoDLA}-\ref{SubFig:MoDLB} for MoDL and in Fig.~\ref{SubFig:CRNN} for CRNN. It can be seen that our algorithm reduces this metric by 4 orders of magnitude for MoDL, and by 9 orders of magnitude for CRNN. Although $\|Aw -b\|_2$ for CRNN is already very small, this does not necessarily translate into good reconstruction performance, as we will see next. 

\begin{table}[t]
    \footnotesize
	\centering
	\renewcommand{\arraystretch}{1}
	\caption{PSNR and SSIM in the format average $\pm$ std, \,min/max. The best values are highlighted in bold.
	}
	\begin{tabular}{@{}l@{\hspace{1.4 \tabcolsep}} l @{\hspace{1.4\tabcolsep}} l@{}}
	\toprule
		\textcolor{thisblue}{\textbf{Method}} & \textcolor{thisblue}{\textbf{PSNR}} & 	\textcolor{thisblue}{\textbf{SSIM}} \\
		\toprule
		{\textcolor{darkgreen}{MoDL}} 
		& 
		39.06 $\pm$ 1.58,\, 33.86/40.91
		& 
		0.97 $\pm$ 0.02,\, 0.84/0.99\\
		\midrule[0.1pt]
		{\textcolor{darkgreen}{Ours} }& 
		\textbf{45.96} $\pm$ \textbf{3.94},\, \textbf{35.48}/\textbf{53.45} 
		& 
		\textbf{0.98} $\pm$ \textbf{0.02},\, \textbf{0.85}/\textbf{1.00} \\
		\toprule
		{\textcolor{myorange}{CRNN}}  
		& 
		24.08 $\pm$ 0.59,\, 22.91/25.29 
		& 
		0.71 $\pm$ 0.03,\, 0.64/0.78
		\\
		\midrule[0.1pt]
		{\textcolor{myorange}{Ours}}  
		& 
		\textbf{25.45} $\pm$ \textbf{0.71},\, \textbf{24.17}/\textbf{26.70} 
		& 
		\textbf{0.76} $\pm$ \textbf{0.02},\, \textbf{0.71}/\textbf{0.80} \\
		\bottomrule
	\end{tabular}
	\label{tab:all}
\end{table}

\mypar{Reconstruction performance}
Table \ref{tab:all} shows the results we obtained on the brain dataset for MoDL, and on the 30 cardiac test set for CRNN. The 2nd (resp.\ 3rd) column displays the average PSNR (resp.\ SSIM) and respective standard deviation over the test images. The first (resp.\ last) two rows refer to the performance of MoDL (resp.\ CRNN) and of the subsequent processing with our method. 

In the case of MoDL, it can be seen that our post-processing increased the PSNR performance by more than 5dB. We also observe only a marginal increase in SSIM. The reason may be because the SSIM values for MoDL were already large and, being constrained to the interval $[0,1]$, were difficult to increase.
However, Figs.~\ref{SubFig:MoDLA}-\ref{SubFig:MoDLB}, which show two examples of test images, demonstrate visually that our method preserves edges, whereas MoDL over-smooths them. The results in Table~\ref{tab:consistency} indicate that this is a by-product of enforcing consistency.

For CRNN, the last two rows of Table~\ref{tab:all} show that our method
improved both the PSNR and SSIM values. The gains, however, were much smaller,
likely because CRNN enforces consistency better than MoDL. The table also shows
that the reconstruction performance using CRNN is much worse than using MoDL,
in part because this network was trained on a single subject.
Fig.~\ref{SubFig:CRNN} demonstrates visually that our method preserves details better.

\section{Conclusions}

We studied the phenomenon of measurement inconsistency in DNNs for linear inverse problems. We achieved this by relating the probability of obtaining an inconsistent output to the generalization error. To overcome this problem, we then proposed a post-processing algorithm that improves the output of DNNs by enforcing consistency.  
Experimental results on MRI reconstruction showed that applying our algorithm not only leads to better consistency, but also to significant reconstruction gains. And the better the improvement in consistency, the larger the gains.


\bibliographystyle{IEEEtran}
{\small
\bibliography{refs}
}

\end{document}